\documentclass[3p,twocolumn,times,number]{elsarticle}

\usepackage{amssymb,amsmath,amsthm}
\usepackage{multirow}
\usepackage{array}
\usepackage{mathrsfs}
\usepackage{mathtools}
\newcommand{\mychi}{\raisebox{0pt}[1ex][1ex]{$\chi$}}
\usepackage{dutchcal}
\usepackage{enumitem}
\usepackage{graphicx}
\usepackage{subcaption}
\usepackage{algorithm,algpseudocode}
\usepackage{tikz}
\usetikzlibrary{positioning,arrows,calc,arrows.meta,patterns,patterns.meta,decorations.markings,decorations.pathreplacing,positioning}
\tikzset{every node/.style={circle}, 
	strike through/.append style={
		decoration={markings, mark=at position 0.5 with {
				\draw[-] ++ (-2pt,-2pt) -- (2pt,2pt);}
		},postaction={decorate}}
}

\setlist[itemize]{noitemsep, topsep=0pt}
\setlist[enumerate]{noitemsep, topsep=0pt}

\journal{Physica A: Statistical Mechanics and its Applications}

\DeclareMathOperator*{\EE}{\textit{E}}
\DeclareMathOperator*{\VV}{\textit{V}}
\DeclareMathOperator*{\E}{W}

\newcommand{\G}{\mathcal{G}}
\newcommand{\HH}{H}
\newcommand{\TT}{T}
\newcommand{\PP}{\mathcal{P}}

\newcommand{\C}{\mathscr{C}}

\newtheorem{theorem}{Theorem}
\newtheorem{lemma}{Lemma} 
\newtheorem{problem}{Problem} 
\newtheorem{corollary}{Corollary} 

\newtheorem{definition}{Definition} 
\newtheorem{proposition}{Proposition} 
\newtheorem{remark}{Remark}

\begin{document}

\begin{frontmatter}

\title{A (Purely) Graph-Theoretic Approach to Synchronization of Nonlinear Dynamical Networks}

\author[1]{Aandrew Baggio Sahaya Arokiadoss\corref{cor1}}
\ead{s.aandrewbaggio@protonmail.com}
\author[2]{G. Arunkumar}
\ead{garunkumar@iitm.ac.in}
\cortext[cor1]{Corresponding author}

\affiliation[1]{organization={Department of Electrical Engineering,
IIT Madras},
addressline={ESB-202},
city={Chennai},
postcode={600036},
state={Tamilnadu},
country={India}}

\affiliation[2]{organization={Department of Mathematics,
IIT Madras},
addressline={KCB1-554},
city={Chennai},
postcode={600036},
city={Chennai},
country={India}}

\begin{abstract}
Synchronizing nonlinear dynamical networks typically requires solving matrix inequalities or detailed system models, which fail for large networks. This paper offers a simple fix : a purely graph-theoretic framework using only a single Lipschitz-like bound on the dynamics. Coupling strengths are computed directly from the digraph, bypassing inequality solvers entirely. The method succeeds where existing approaches encounter infeasibility due to connectivity patterns. It examines only $n-1$ directed paths per strongly connected component versus $\frac{n(n-1)}{2}$ undirected paths before, achieving $O(n^3)$ complexity. Results show network connectivity can be exploited to synchronize a large class of nonlinear dynamical networks.
\end{abstract}


\end{frontmatter}

\section{Introduction}
While linear network control benefits from a well-developed theoretical foundation, nonlinear network synchronization is comparatively less mature. This is due to the diverse collective behaviours they may exhibit, even when the interconnection patterns are simple. In biology, for instance, each system in the network may represent a biochemical or genetic subsystem where the number of states by itself is huge and the dynamics are intricate \cite{mirsky2009model,Oda2005} and often unknown. In such settings, developing complete mathematical models is so challenging that model construction itself becomes a substantial task worthy of appreciation.

This difficulty has led researchers to investigate whether synchronization can be achieved primarily through knowledge of the network's connectivity, with only limited information about the underlying system dynamics. A major conceptual step in this direction was taken by Mochizuki and Fiedler in their two part series on the graph theoretic analysis of complex biochemical systems, \cite{fiedler2013dynamics,mochizuki2013dynamics}. The authors made two key assumptions: general dissipative behavior and the signs of the self-derivatives are known. They reformulated controlling the long-term behavior of a regulatory network as a classic problem of finding the minimal feedback vertex set. These works captured the idea that the essential behaviour of a network can be understood through its connectivity. Mochizuki's work offers valuable structural insights but lacks explicit guidance on selecting inputs or coupling strengths. 

Lyapunov's method addresses this gap effectively. It leads to matrix inequalities, which, when solved, provided concrete conditions for inputs and coupling gains to guarantee synchronization. On the other hand, works like \cite{belykh2006generalized,belykh2004connection,liu2015synchronization}, utilized spectral graph theory to reduce the resulting matrix inequality into a set of linear inequalities. These approaches identify suitable paths in the network's directed graph and use these paths to assign coupling strengths that meet the linear inequalities. However, these methods face infeasibility issues. Inequality solvers may fail to find solutions if some parameters are chosen carelessly. They also require computing $\frac{n(n-1)}{2}$ path to determine coupling strengths. This becomes a computational nightmare as size of the network increases.

Motivated by these insights, we propose our method, which is consistent with the philosophy in \cite{fiedler2013dynamics,mochizuki2013dynamics} and offers the following contributions :
\begin{enumerate}
    \item \textbf{Purely graph-theoretic framework:} The method relies on standard graph algorithms and eliminates the need for detailed system models or expensive numerical solvers.
    \item \textbf{Resolves infeasibility in existing methods:} The method guarantees the existence of coupling strengths for synchronizing networks obeying the stated assumptions.
    \item \textbf{Reduced computational complexity:} For each strongly connected component with $n'$ vertices, our method needs just $n'-1$ directed paths and $n'-m+1$ directed cycles (where $m$ is the number of arcs). This yields $O(n^3)$ overall complexity, scaling well to large networks with $n$ vertices.
\end{enumerate}
\section{Preliminaries}
\subsection{Graphs and Digraphs}
A \textit{digraph} is an ordered pair $(\VV, \EE)$. $\VV = \{v_1, \dots, v_n\}$ is its set of vertices and $\EE = \{e_1, \dots, e_m\} \subseteq \VV \times \VV$ is its set of arcs. An \textit{undirected graph} (or graph) is an ordered pair $(\VV, \EE)$, where $\VV$ is defined the same and $\EE$ is a multiset of unordered vertex pairs.  
Throughout this work, we consider \textit{simple digraphs}, i.e., digraphs without self-loops or parallel arcs. Graphs may include parallel edges, but a \textit{simple graph} contains neither. An arc in a digraph (or edge in a graph) can be referred either by the vertices it connects ($e_{ij}$) or by its label number $e_{k}$. We use the same notations to represent the weights associated with the arcs (or edges) i.e. $w_{k}$ or $w_{ij}$.
\begin{table}
\caption{Functions, terms and definitions}
\label{tab:Graph_definitions}
\begin{tabular}{lll}
\hline
\textbf{Term/Function} &&\textbf{Definition/Description}\\
\hline
$e_{ij} = (v_i, v_j)$ 
  &&Arc from $v_i$ (tail) to\\
  &&$v_j$ (head)\\
$e_{ij} = \{v_i, v_j\}$ 
  &&Edge with $v_i$ and $v_j$\\
  &&as its vertices\\
End vertices 
  &&Vertices of an arc\\
  &&or edge\\
Outgoing (Incoming) arc&&Arc with the\\
&&vertex as its tail\\
&&(or its head)\\
Incident arcs (edges)
  &&Arcs (edges) having the\\
  &&vertex as one of their\\  
  &&end vertices\\
Degree 
  &&Number of incident edges\\
Outdegree (Indegree)
  &&Number of outgoing\\
  &&(or incoming) arcs\\
Adjacent arcs 
  &&Arcs sharing at least\\
  &&one end vertex\\  
Source vertex 
  &&Vertex with $0$ indegree\\
Parallel arcs 
  &&Arcs with the same\\
  &&head and the same tail\\
Parallel edges 
  &&Edges with the same\\
  &&end vertices\\
Self-loops 
  &&Arcs with identical end\\
  &&vertices\\
Underlying graph 
  &&Graph formed by\\
  &&replacing each arc with\\
  &&an edge having the\\
  &&same end vertices\\
Traversal 
&&An ordered sequence of\\
&&vertices where each\\
&&consecutive pair\\
&&represents the endpoints\\
&&of an arc\\
$\VV(\G)$ 
  &&Vertex set of $\G$\\
$\EE(\G)$ 
  &&Arc (or edge) set of $\G$\\
$|\G|$ 
  &&Number of arcs in $\G$\\
$\HH$ 
  &&Maps an arc to its head\\
$\TT$ 
  &&Maps an arc to its tail\\
  \hline
\end{tabular}
\end{table}
A \textit{weighted digraph} is an ordered triplet $(\VV, \EE, \E)$, where $\E = [w_1, \dots, w_m]^\top \in \mathbb{R}^m$ is the arc-weight vector where $w_i$ is the weight associated with the arc $e_{i}\in \EE$.

\begin{definition}[Vertex Imbalance]
In a weighted digraph $\G$, the vertex imbalance $\text{d}_i$ of a vertex $v_i\in \VV(\G)$, is the difference between the sum of its outgoing and incoming arc weights. 
\[\text{d}_i = \qquad
\smashoperator{\sum_{\left\{j:\substack{(v_i,v_j) \in \EE(\G)}\right\}}} w_{ij}\quad
- \quad\smashoperator{\sum_{\left\{\substack{k:(v_k,v_i) \in \EE(\G)}\right\}}} w_{ki}\]
\end{definition}
\subsection{Undirected subgraphs}
A \textit{subgraph} $\mathcal{H}$ of a digraph (graph) $\G$ is a digraph (graph) such that $\VV(\mathcal{H}) \subseteq \VV(\G)$ and $\EE(\mathcal{H}) \subseteq \EE(\G)$, and is denoted as $\mathcal{H} \subseteq \G$.  
A \textit{path} $\PP$ is a graph with $\VV(\PP) = \{v_{k_0}, \dots, v_{k_\ell}\}$ and $\EE(\PP) = \{e_{k_1}, \dots, e_{k_\ell}\}$, where $e_{k_i} = \{v_{k_{i-1}}, v_{k_i}\}$. The \textit{end vertices} of a path graph are its vertices with degree 1. Two vertices in a graph are said to be \textit{connected} if the graph contains a path subgraph that has the vertices as its end vertices. A graph is \textit{connected} if every pair of its distinct vertices is connected.  
A \textit{cycle} is a connected graph in which every vertex has degree two. A \textit{tree} is a connected graph with no cycle subgraph, and a \textit{spanning tree} is a tree subgraph that includes all the vertices of the graph.  
An $n$ vertex \textit{complete graph}, denoted by $\mathbb{K}_n$, is a graph in which every vertex is adjacent to every other vertex. A \textit{star graph} is a graph with one central vertex adjacent to all the other vertices, while the remaining vertices are adjacent only to the central vertex.
\subsection{Directed subgraphs}
A \textit{directed path} $\PP$ satisfies $\VV(\PP) = \{v_{k_0}, \dots, v_{k_\ell}\}$ and $\EE(\PP) = \{e_{k_1}, \dots, e_{k_\ell}\}$ with $e_{k_i} = (v_{k_{i-1}}, v_{k_i})$. Its tail and head are $\TT(\PP) = \TT(e_{k_1})$ and $\HH(\PP) = \HH(e_{k_\ell})$ and we say that there is a directed path from $\TT(e_{k_1})$ to $\HH(e_{k_\ell})$. A digraph is connected if its underlying graph is connected. A \textit{directed cycle} is a digraph with at most one arc between any of its vertex pairs and its underlying graph is a cycle.  
A \textit{directed spanning tree} is a digraph that has a single source vertex with a unique directed path from it to every other vertex. The \textit{length} of a directed path or cycle equals its number of arcs. A digraph is said to be \textit{strongly connected} if, for every ordered pair of distinct vertices, there exists a directed path from the first to the second.
\subsection{Matrices and Vectors}
For $\text{X} \in \mathbb{R}^n$, $\text{X} > 0$ (resp. $\ge 0$) indicates all entries are positive (resp. nonnegative). The zero and all-one vectors are $\mathbf{0}_n$ and $\mathbf{1}_n$ respectively. For a weighted digraph $\G$ with $n$ vertices and $m$ arcs, the incidence ($\text{Q}$) is defined as :
\begin{equation*}
\begin{aligned}
[\text{Q}]_{ij} &=
\begin{cases}
\phantom{-}1, & \text{ if } v_i = \TT(e_j),\\
-1, & \text{ if } v_i = \HH(e_j),\\
\phantom{-}0, & \text{otherwise,}
\end{cases}
\end{aligned}
\end{equation*}
\begin{definition}
For a weighted digraph \(\G = (\VV, \EE, \E)\) with incidence matrix Q, the vertex imbalance vector D contains the vertex imbalance of each vertex as its entries and satisfies Q$\E$=D.
\end{definition}
Let $w_{ij}$ be the arc weight corresponding to the arc $(v_i,v_j)$. The corresponding Laplacian matrix (L) is given by :
\begin{equation*}
\begin{aligned}
[\text{L}]_{ij} &=
\begin{cases}
\phantom{\quad \sum}-w_{ji}, &\text{ if } i \neq j \text{ and } (j,i) \in \EE(\G),\\
\;\;\quad\smashoperator{\sum_{\left\{j:(v_j,v_i)\in \EE(\G)\right\}}}w_{ji}, &\text{ if } i=j,\\
\phantom{-\qquad \sum}0, & \text{otherwise.}
\end{cases}
\end{aligned}
\end{equation*}
$\text{L}_{\G}$ denotes the Laplacian of $\G$ with $m=|\EE(\G)|$. In unweighted cases, all the weights are assumed to be equal to $1$. Fixing a traversal for each directed cycle $\mathcal{H}$ in $\G$, results in a corresponding \textit{signed incidence vector} $\text{X} = [x_1, \dots, x_m]^\top$ :
\begin{equation*}
x_k =
\begin{cases}
\phantom{-}1, & \text{if } e_k \in \EE(\mathcal{H}) \text{ tail precedes head},\\
-1, & \text{if } e_k \in \EE(\mathcal{H}) \text{ head precedes tail},\\
\phantom{-}0, & \text{if } e_k \notin \EE(\mathcal{H}).
\end{cases}
\end{equation*}
 
The \textit{cycle space} $\C$ of $\G$ is the nullspace of Q i.e., $\C = \{\text{X} \in \mathbb{R}^m : \text{QX} = \mathbf{0}_n\}$. A \textit{cycle basis} is a set of linearly independent signed incidence vectors spanning $\C$ and each vector in it corresponds to a distinct directed cycle of $\G$.
\section{Problem Formulation}
We consider a network of $n$ identical, diffusively coupled dynamical system following \cite{liu2015synchronization} :

\begin{equation}
\dot{\boldsymbol{z}}_{i} = f(\boldsymbol{z}_{i}) + \smashoperator{\sum_{\left\{j:(v_j,v_i)\in\EE(\G)\right\}}} w_{ji}\,\text{P}\boldsymbol{z}_{j} + u_{i}
\label{eq:Dynamical Network Model}
\end{equation}
Here, $f$ is the system dynamics, $\boldsymbol{z}_i \in \mathbb{R}^d$ the state, $w_{ji}$ is the coupling strength with which system $j$ influences system $i$ and $\text{P}$ is a $\{0,1\}$ diagonal matrix indicating which state components are coupled. $u_{i}$ is the input given to the system $i$. The couplings and their strengths form a weighted digraph $\G$, termed the \textit{connectivity digraph}, where the arc weights represent coupling strengths. The synchronization manifold of~\eqref{eq:Dynamical Network Model} is defined as the set $\left\{\, \mathbf{z} \;\middle|\; \mathbf{z} = \mathbf{1}_n \otimes \boldsymbol{z},\ \boldsymbol{z} \in \mathbb{R}^d \,\right\}$. Here, $\otimes$ represents the Kronecker product. The network is said to be \emph{synchronized} if the states of all the systems in it asymptotically converge to the trajectories within this set. Global stability of this manifold is ensured under the following assumptions \cite{belykh2004connection}:
\begin{enumerate}
    \item For a choice of P, $\exists \, a>0$ such that $\forall\, \boldsymbol{x},\boldsymbol{y} \in \mathbb{R}^d$, and some $c>0$,
    \begin{equation}
    (\boldsymbol{x} - \boldsymbol{y})^{\top}[f(\boldsymbol{x}) - f(\boldsymbol{y}) - a\text{P}(\boldsymbol{x} - \boldsymbol{y})] \leq -c||\boldsymbol{x}-\boldsymbol{y}||^{2}
    \label{eq:Synchronization Assumption}
    \end{equation}
    This assumption is equivalent to contraction assumption \cite{delellis2010quad} implying monotonically decreasing Master Stability function (MSF) \cite{russo2009contraction}. Note that all oscillators do not obey this assumption with a prime example being x-coupled Rossler's attractors.
    \item The connectivity digraph contains a directed spanning tree.
\end{enumerate}

The objective is to find an arc-weight vector $\E$ such that 
\[
\lim_{t \to \infty} \|\boldsymbol{z}_i(t) - \boldsymbol{z}_{j}(t)\| = 0,\quad \forall\, i,j=1,\dots,n.
\]
The Lyapunov-based approach leads to the following matrix inequality \cite[Theorem~4.4]{wu2007synchronization}:
\begin{equation}
(\text{U} \otimes \text{V}) \big( \text{L}_{\G} \otimes (-\text{P}) - \text{I}_{n} \otimes \text{Y} \big) \preceq 0,
\label{eq:Synchronization Matrix Inequality}
\end{equation}
where U is an irreducible symmetric matrix with non-positive non-diagonal elements and V is a symmetric positive definite matrix. This equation was simplified in \cite{liu2013coupling} using the assumption 1 and several other assumptions  to be :
\begin{equation}
\text{L}_{\G} \text{L}_{\G_0} - a \text{L}_{\G_0} \succeq 0.
\label{eq:Graph inequality}
\end{equation}
Here,  L$_{\G_0}$ is the Laplacian matrix of a connected undirected graph with the same vertex set as $\G$. For $\text{L}_{\G_0}=\text{L}_{\mathbb{K}_n}$, the synchronization condition becomes
\begin{equation}
w_{ij}^{s} > \frac{a}{n}\left(\;\;\;\;\quad \smashoperator{\sum_{\left\{\PP \in \mathcal{S}: e_{ij} \in \EE(\PP)\right\}}} \Omega(\PP)\,\right),\;\forall\, i,j=1,\dots,n\: (i \neq j)
\label{eq:Synchronization Condition}
\end{equation}
Let $e_{ij}\in\EE(\G)$. Now, $w_{ij}^{s}=\dfrac{w_{ij}+w_{ji}}{2}$ if $e_{ji}\in\EE(\G)$, and $w_{ij}^{s}=\dfrac{w_{ij}}{2}$ otherwise. Also, $\mathcal{S}$ denotes the set of undirected paths in the underlying graph of $\G$, one for each unordered pair of distinct vertices and having the same as its end vertices. $\Omega(\PP)$ is the weight associated with each such path, defined as :
\[
\Omega(\PP) =
\begin{cases}
|\PP|\;\mychi\!\left(1+\dfrac{d_{\HH(\PP)}+d_{\TT(\PP)}}{2a}\right) & \text{, if }e_k \notin \EE(\G)\\
1+\dfrac{d_{\HH(e_k)}+d_{\TT(e_k)}}{2a}&  \text{, otherwise }
\end{cases}
\]
where $\mychi(z)=\max\{z,0\}$. Having outlined the method in \cite{liu2015synchronization}, we next identify its key issues.

\subsection{Infeasibility in Existing Method}
The generalized connection graph method (GCGM)~\cite{liu2015synchronization} employs spectral graph theory to determine coupling strengths that satisfy the synchronization condition~\eqref{eq:Synchronization Condition}. The key idea is that if a dynamical system receives inputs from multiple identical synchronized systems, these can be equivalently represented by a single synchronized copy with appropriate scaling. This is because the influencing terms of all the systems are the same~\cite[Sec.~IV]{liu2015synchronization}. Consequently, each strongly connected component (SCC) of the connectivity digraph can be collapsed into a single vertex. Extending this reasoning, all SCCs that have outgoing arcs to a given SCC can be merged into one auxiliary vertex (see Fig.~\ref{fig:connectivity_digraph}). This reduction allows the GCGM algorithm to process one SCC at a time, significantly lowering computational complexity. Any resulting parallel arcs are then combined into a single arc, and the arc weights are computed for the reduced digraph. This view leads us to the following remark :
\begin{remark}
\label{thm:SCC with 1 or 2 components}
GCGM deals with two types of digraphs :
\begin{enumerate}
\item strongly connected, or
\item composed of two SCCs, one being a single source vertex.
\end{enumerate}
\end{remark}
Nonetheless, the GCGM algorithm can end up with infeasible solutions for certain connectivity digraphs. To illustrate, consider a digraph composed of two SCCs, as characterized in Remark~\ref{thm:SCC with 1 or 2 components}. Suppose that the non-trivial SCC constitutes a directed cycle in which only a single vertex receives incoming arcs originating from the trivial SCC. Let $d_{1}$ be the imbalance of the source vertex, -$d_{1}$ be the imbalance of its adjacent vertex, and $0$ for all others. For this selection of vertex imbalances, the only possible choice of arc weight distribution is to keep all the arc weights in the non-trivial SCC to be equal. Using the shortest path between any two vertices and \eqref{eq:Synchronization Condition}, the synchronization inequality for $w_{01}$ becomes
\begin{equation}
w_{01} >
\begin{cases}
\dfrac{2a}{n}\!\left(\text{P}^{o}_{\text{sum}} + (\text{P}^{o}_{\text{sum}} - 1)\dfrac{w_{01}}{2a}\right), & \text{ if } n~\text{ is odd},\\[6pt]
\dfrac{2a}{n}\!\left(\text{P}^{e}_{\text{sum}} + (\text{P}^{e}_{\text{sum}} - 1)\dfrac{w_{01}}{2a}\right), & \text{ if } n~\text{ is even}.
\end{cases}
\end{equation}
Here, $\text{P}^{o}_{\text{sum}} = \dfrac{n^{2} + 2n - 3}{4}$ and $\text{P}^{e}_{\text{sum}} = \dfrac{n^{2} + 2n - 4}{4}$. 
For $n > 4$, the inequality yields only non-positive solutions, rendering positive arc weights infeasible. 
This behaviour arises from the fact that the numerator is a quadratic polynomial of $n$, whereas the denominator varies linearly with $n$. Feasibility can be restored by assigning vertex imbalances such that the weighted path-sums grow linearly (or slower) with $n$. This a critical issue is overlooked in \cite{liu2015synchronization}.

\begin{figure}
\centering
\begin{tikzpicture}[>=stealth,scale=0.3,
    redvertex/.style={circle, draw, fill=none, pattern color=black!50!red, minimum size=14pt, inner sep=0pt},
    bluevertex/.style={circle, draw, fill=blue!30, minimum size=14pt, inner sep=0pt},
    greenvertex/.style={circle, draw, fill=green!30, pattern=north west lines, pattern color=black!40!green, minimum size=14pt, inner sep=0pt},
    arrow/.style={thick},
    node distance=5mm]

\node[bluevertex] (v0) at (90:8cm) {$0$};

\foreach \i/\angle in {
        1/90,
        2/18,
        3/-54,
        4/-126,
        5/162}
    \node[redvertex] (v\i) at (\angle:4cm) {$\i$};

\draw[->,arrow] (v0) -- node[midway, right] {$e_{01}$} (v1);

\draw[->,arrow] (v1) -- node[midway, above right] {$e_{12}$} (v2);
\draw[->,arrow] (v2) -- node[midway, right] {$e_{23}$} (v3);
\draw[->,arrow] (v3) -- node[midway, below] {$e_{34}$} (v4);
\draw[->,arrow] (v4) -- node[midway, left] {$e_{45}$} (v5);
\draw[->,arrow] (v5) -- node[midway, above left] {$e_{51}$} (v1);
\end{tikzpicture}
\caption{A digraph with 6 vertices that causes feasibility issues in inequality solvers approach}
\label{fig:counterexample}
\end{figure}
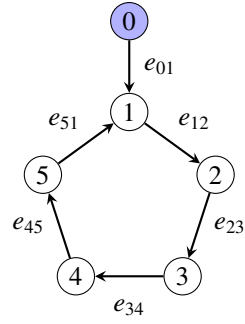
\textit{Example} : Consider a digraph as shown in Fig.~\ref{fig:counterexample}. Let $w_{01}$, $w_{12}$, $w_{23}$, $w_{34}$, $w_{45}$ and $w_{5,1}$ be the arc weights of $e_{01}$, $e_{12}$, $e_{23}$, $e_{34}$, $e_{45}$ and $e_{51}$ respectively. Let us assume an arbitrary vertex imbalance vector $\left[\;d_{1}\,-d_{1}\,0\,0\,0\,0\;\right]$ with the arc-weight vector. For $w_{01}$, we get the following inequality :
 \[
 w_{01}>\dfrac{2a}{6}\left(\qquad\;\smashoperator{\sum_{\;\;\left\{\PP \in \mathcal{S}: e_{01} \in \EE(\PP)\right\}}} \Omega(\PP)\;\right)
 \]
If $\mathcal{S}$ is a set of shortest paths between any two vertices in the given digraph, we get P$^{e}_{sum}=11$.
\begin{equation*}
w_{01}>\frac{a}{3}\Big(11 + 10\frac{w_{01}}{2a}\Big)\implies w_{01}<-\frac{11a}{2}
\end{equation*}
This is impossible because $w_{01}$ must be strictly positive by construction. Hence the previous algorithm, when implemented with the simple equal-weight parameterization, admits no feasible positive solution for this digraph. This demonstrates the infeasibility phenomenon highlighted earlier.

\subsection{Increase in the number of path computations}
For a graph with $n$ vertices, the GCGM method requires computing $\dfrac{n(n-1)}{2}$ undirected paths. This quadratic growth in path computations becomes prohibitively expensive for large networks.

\subsection{Problem Statement}
Our objective now is to formulate an approach that removes the limitations stated and to synchronize any dynamical network satisfying the stated assumptions. To minimize the number of paths to be computed, we employ the following observation from \cite{liu2015synchronization} :
\begin{remark}
For any path $\PP$ in $\G$ whose end vertices have imbalance $\leq -a$, it follows that $\mychi(\omega(\PP)) = 0$.
\label{thm:Negative Vertex Imbalance}
\end{remark}As a result of the above remark, paths with  $\mychi = 0$ can be disregarded in evaluating \eqref{eq:Synchronization Condition}. Consequently, ensuring that most vertex imbalances remain $\leq -a$ decreases the number of contributing paths.
Now, we state the problem :
\begin{problem}
Given $a > 0$ and a weighted digraph $\G = (\VV, \EE, \E)$ that is either strongly connected or composed of a source vertex and a non-trivial SCC, find a $\E > 0$ such that:
\begin{enumerate}
\item The synchronization condition \eqref{eq:Synchronization Condition} holds and
\item All but one vertex in the non-trivial SCC have imbalance $\leq -a$.
\end{enumerate}
\end{problem}

\section{Coupling Strength Allocation}
In weighted digraphs restricted to positive arc weights, the number of vertices that can have negative vertex imbalances is at most $n-1$. This upper bound is due to the handshaking lemma.
\begin{lemma}[Handshaking lemma for weighted digraphs]
\begin{equation*} \sum_{v_{i}\in \VV(\G)} \Bigg(\quad\smashoperator{\sum_{\;\;\;\;\;\;\{j: (v_{i},v_{j})\in \EE(\G)\}}}w_{ij}\quad-\quad\smashoperator{\sum_{\;\;\;\;\;\;\;\;\;\;\{k: (v_{k},v_{i})\in \EE(\G)\}}}w_{ki}\quad\Bigg)=0 \iff \smashoperator{\sum_{v_{i}\in \VV(\G)}} d_{i} = 0
\end{equation*}
\label{lem:handshake}
\end{lemma}
We define an appropriate arc-weight vector to achieve this upper bound for a weighted digraph
\begin{definition}[Negative Imbalance Arc-weight vector]
A \textit{negative imbalance arc-weight vector} is an arc weight vector with positive entries that results in negative vertex imbalances at $n-1$ vertices of an $n$-vertex digraph.
\end{definition}
Let X$^{-}$ denote the negative imbalance arc–weight vector of the weighted digraph $\G$. Owing to the linearity of the relation QX$^{-} =$D , scaling X$^{-}$ appropriately ensures that all but one negative vertex imbalances become $\leq -a$. Consequently, the value of the function $\chi$ becomes zero for all paths except those originating from the single vertex with positive vertex imbalance. Hence, only $n-1$ directed paths need to be computed. The problem therefore reduces to determining a suitable negative imbalance arc–weight vector for $\G$.

To address this systematically, we split the arc–weight allocation into intra-SCC and inter-SCC allocation. For intra-SCC arcs, that is, arcs whose end vertices lie in the same SCC, we show that a negative imbalance arc–weight vector can always be constructed. For inter-SCC allocation, we find X$^{-}$ for the non-trivial SCC. We use this vector and the Laplacian matrix of a star graph. This to assign weights to arcs whose head vertex lies in the SCC under consideration. Finally, we prove that the resulting arc weights satisfy condition~\eqref{eq:Graph inequality}.

\subsection{Intra-SCC Allocation}
We now establish that any strongly connected weighted digraph admits a negative imbalance arc-weight vector.

\begin{proposition}
\label{thm:maximum negative vertex imbalance vector}
For any directed $\PP$ of a weighted digraph $\G$, there exists an arc-weight vector that can result in the vertex imbalances satisfying the following : 
\begin{enumerate}
    \item $\text{d}_i > 0$ if $v_i = \TT(\PP)$,
    \item $\text{d}_i < 0$ for $v_i \in \VV(\PP) \setminus \{\TT(\PP)\}$,
    \item $\text{d}_i = 0$ for $v_i \notin \VV(\PP)$.
\end{enumerate}
\end{proposition}
\begin{proof}
Consider a directed path $\PP$ of length $\ell$ with arc sequence 
$e_1, \dots, e_\ell$. Assign weights such that $0 < w_1 < w_2 < \dots < w_\ell,$
where $w_i$ is the weight of $e_i$. Assign weight $0$ to every arc not belonging to $\PP$. By construction, only the vertices on $\PP$ have nonzero imbalance. This gives the vertex imbalance vector : 
\begin{equation*}
   \text{D} = [\, w_1,\; w_2 - w_1, \dots, -w_\ell,\; 0,\dots,0\,]^{\top} 
\end{equation*} The first entry is positive and the next $\ell-1$ entries are negative, thus proving the theorem.
\end{proof}

\begin{corollary}
\label{thm:negative imbalance arc-weight vector for strongly connected digraphs}
Every strongly connected digraph admits a negative imbalance arc-weight vector.
\end{corollary}
\begin{proof}

Choose an arbitrary vertex $v_r$ as the root and construct a directed spanning tree. 
For each vertex $v_j \neq v_r$, there exists a directed path from $v_r$ to $v_j$.

By Proposition~\ref{thm:maximum negative vertex imbalance vector}, 
each such path defines a vector X$^{-}_j$ that produces a vertex–imbalance vector 
satisfying
\[
\text{d}_r > 0, 
\qquad 
\text{d}_j < 0, 
\qquad 
\text{and} 
\qquad 
\text{d}_k \le 0 \quad \forall \, k \neq r,j.
\]
Summing these vectors over all $j \neq r$ preserves $\text{d}_r > 0$, 
since each X$^{-}_j$ contributes positively at $v_r$. 
For every $i \neq r$, at least one term makes $\text{d}_i < 0$, 
and no term contributes positively at $v_i$. 
Hence, in the resulting vector,
\[
\text{d}_r > 0 
\quad \text{and} \quad 
\text{d}_i < 0 \quad \forall \, \neq r.
\]
This proves the claim.
\end{proof}

We use a weight difference of $1$ for constructing the vector for each directed path in proposition~\eqref{thm:maximum negative vertex imbalance vector}. We scale the summation of these vectors i.e., X$^{-}$ by $a$ to ensure any path not starting at $v_r$ has a negative weight i.e., $\chi < 0$. Here, X$^{-}$ is a negative imbalance arc-weight vector and it changes the synchronization condition in \eqref{eq:Synchronization Condition} to be \begin{equation}
w_{k}^{s} \ge \frac{2a}{n}
\left(1 + \sum_{\PP \in \mathcal{S}^{'}} |\PP|\right)
\left(\sum_{\PP \in \mathcal{S}^{'}} |\PP|\right)
\label{eq:Modified Synchronization Condition}
\end{equation}
Here, $\mathcal{S}'$ denotes a set of directed paths, one for each vertex other than $v_r$. 
Each path starts at $v_r$ and ends at a distinct vertex.

Although the number of paths to be computed is reduced, the synchronization condition~\eqref{eq:Modified Synchronization Condition} is not yet satisfied. 
To enforce it, the arc weights must be increased till all the inequalities are satisfied without altering the vertex imbalances. 
This is achieved by adding a vector from the cycle space of the SCC (X$^{0} \in \C$). 
Since X$^{0}$ lies in the cycle space, the right-hand side of~\eqref{eq:Synchronization Condition} remains unchanged, while the left-hand side can be increased as much as we need.

We construct X$^{0}$ as the sum of directed-ear basis vectors. A directed-ear basis is a cycle basis obtained via directed-ear decomposition~\cite[Section~2.1]{loebl2001some}. Each directed ear is a directed path whose internal vertices have in-degree and out-degree equal to $1$. Moreover, the corresponding signed incidence vectors contain non-negative entries under at least one traversal orientation. If the considered SCC is a source SCC, we scale X$^{0}$ by a factor $\Delta w$ that is larger than the highest value on the right hand side of \eqref{eq:Modified Synchronization Condition}.

\begin{equation}
\Delta w = \frac{2a}{n}
\left(1 + \sum_{\PP \in \mathcal{S}^{'}} |\PP|\right)
\left(\sum_{\PP \in \mathcal{S}^{'}} |\PP|\right)
\label{eq:X0scaling}
\end{equation}
For non-source SCCs, the scaling factor becomes $1$ since the synchronization condition is taken care of by the inter-SCC incoming arc weights. We define $\text{W}^{0}=\Delta w\,\text{X}^{0}$ to be the scaled vector and add it to X$^{-}$ to yield the desired arc weight vector for the SCC considered.
\[
\text{W}_\text{sync} =\text{W}^{-} + \text{W}^{0},
\]
We detail this procedure in algorithm~\ref{alg:intra-scc}.

\begin{algorithm}[h]

\caption{Intra-SCC Arc Weight Allocation}
\label{alg:intra-scc}
\begin{algorithmic}
\State \textbf{Input:} $\G_{k}$ (SCC of $\G$), $v_r$ (root vertex ), $a>0$ (system parameter) and $\Delta w$ (scaling factor).
\State \textbf{Output:} Arc-weight vector $\E_{\text{sync}}$.
\State \textbf{Step 2:} $\E^{-} \gets \mathbf{0}_{| \EE(\G_{k})|}$ and $n\gets |\VV(\G_{k})|-1$
\State \textbf{Step 2:} Assign $v_{r}$ as vertex 1 i.e. $v_{1} \gets v$.
\State \textbf{Step 3:} $\mathcal{S} \gets \{\mathcal{P}_{j}|\mathcal{P}_{j} \text{ starts at }v_{1}\text{ and ends at } v_{j}\}$ (one path per vertex).
\State \textbf{Step 5:} Update the weights of the arcs in $\mathcal{P}_{j}$ : \\ \hspace{1.126cm}$\E^{-}(e_{i}) \gets \E^{-}(e_i) + (\ell - i + 1)$.\\ \hspace{1.124cm} ($i$ is the path sequence number).
\State \textbf{Step 6:} $n \gets n-1$.
\State \textbf{Step 7:} \textbf{if} $n>0$ return to Step 4.
\State \textbf{Step 8:} Scale the vector: $\E^{-} \gets a\E^{-}$.
\State \textbf{Step 9:} Compute directed ear basis for $\G_{k}$:\\ \hspace{1.127cm}$\mathcal{E}\gets \{\text{X}_1,\dots,\text{X}_k\}$
of $S$.
\State \textbf{Step 10:} $\E^{0} \gets$ $\Delta w\sum_{ \text{X}\in \mathcal{E}} \text{X}$.
\State \textbf{Step 11:} $\E_{\text{sync}} \gets \E^{-} + \E^{0}$
\State \textbf{Step 12:} Return $\E_{\text{sync}}$.
\end{algorithmic}
\end{algorithm}

\subsection{Inter-SCC Arc Allocation}
Due to the remark~\eqref{thm:SCC with 1 or 2 components}, it is sufficient to establish the synchronization condition for a digraph consisting of two SCCs, with one being the a source vertex. The weight assigned to each arc from this source vertex is then distributed evenly among the corresponding inter-SCC incoming arcs that it represents.
 
 We propose the following theorem to ensure that the synchronization condition~\eqref{eq:Graph inequality} is satisfied for the reduced digraph as mentioned in Remark~\ref{thm:SCC with 1 or 2 components} :
\begin{theorem}
\label{thm:Arc weights for digraphs with trivial SCCs}
A weighted connectivity digraph with exactly two strongly connected components, one of which is a source vertex, admits a strictly positive arc - weight vector satisfying the synchronization condition in \eqref{eq:Graph inequality}.
\end{theorem}
\begin{proof}
Let $\text{L}_{\G}$ denote the Laplacian matrix of the weighted digraph described in the statement, and let $\text{L}_{\G_0}$ represent the Laplacian of an unweighted star graph with the same vertex set as $\G$ i.e., $\VV(\G)=\VV(\G_0)$. The central vertex of $\G_0$ is assigned the same label as the source vertex of $\G$.

Choose an arbitrary vertex that is adjacent to $v_{0}$ in $\G$ as $v_{1}$.  
By Corollary~\ref{thm:negative imbalance arc-weight vector for strongly connected digraphs}, the non-trivial SCC of $\G$ admits a negative imbalance vector X$^{-}$ such that $v_1$ has a positive imbalance.  
Let L$_{\text{non-triv}}$ denote the Laplacian matrix of this non-trivial SCC, and define its imbalance vector as D$ = [\,d_1, \dots, d_{n-1}\,]^{\top}$.  
Denote the set of outgoing arcs from $v_0$ as $\{e_{0,1}, \dots, e_{0,m_1}\}$.  
Correspondingly, define the vector W$_{\text{triv}} = [\,w_{0,1},\dots,w_{0,m_1},\,\boldsymbol{0}_{n-m_1-1}^{\top}]^{\top}$, which collects the weights of the outgoing arcs from the source vertex followed by zeros for the remaining arcs.
The inequality in \eqref{eq:Graph inequality} can be equivalently expressed in a symmetric form, given by:
\begin{equation}
\frac{1}{2}\left(\text{L}_{\G_0}\text{L}_{\G} + \text{L}_{\G}^\top \text{L}_{\G_0}\right) \succeq a\,\text{L}_{\G_0}.
\label{eq:symmetric-form of Graph inequality}
\end{equation}
To simplify the left-hand side of \eqref{eq:symmetric-form of Graph inequality}, we add and subtract the diagonal matrix 
$\mathrm{diag}\!\left([0\;\;\frac{1}{2}\text{D}]\right)$, yielding the decomposition:
\begin{equation}
\begin{bmatrix}
0 & \mathbf{0}_{n-1}^\top \\[2pt]
\mathbf{0}_{n-1} &  \text{L}^{\text{sym}}_{\text{non-triv}}
\end{bmatrix}
+
\begin{bmatrix}
d_0 & \dfrac{1}{2}\text{D}^\top - \E_{\text{triv}}^\top \\[2pt]
\dfrac{1}{2}\text{D} - \E_{\text{triv}} & \text{diag}(\E_{\text{triv}} - \dfrac{1}{2}\text{D})
\end{bmatrix},
\label{eq:split inequality}
\end{equation}
where 
\[
\text{L}^{\text{sym}}_{\text{non-triv}} 
= \frac{1}{2}\!\left(\text{L}_{\text{non-triv}} + \text{L}_{\text{non-triv}}^\top + \text{diag}(\text{D})\right)
\]
The first term in \eqref{eq:split inequality} , has zero row sum, and each of its diagonal entries is larger than the sum of the absolute values of the non-diagonal elements in the corresponding row.  
This property follows from the addition of the term $\mathrm{diag}\!\left([0\;\;\frac{1}{2}\text{D}]\right)$.  
According to the Gershgorin circle theorem, this matrix is positive semi-definite.  
So, we can focus solely on the second term in \eqref{eq:split inequality}.
The second matrix term in \eqref{eq:split inequality} :
\begin{enumerate}
    \item is symmetric with non-negative off-diagonal entries,
    \item has the same zero non-zero pattern as the Laplacian of an undirected star graph, and
    \item its row sum is zero.
\end{enumerate}
These properties imply that the second matrix on the left-hand side can be interpreted as the Laplacian of an undirected weighted star graph. Since the matrix on the right-hand side of the inequality is scaled by \( a \), it corresponds to a weighted star graph in which each edge has a uniform weight \( a \). Therefore, it suffices to ensure that every edge weight of the star graph represented on the left-hand side is greater than or equal to \( a \).
Since all imbalance terms except $d_{1}$ are negative, we get the conditions :
\begin{equation}
    \begin{aligned}
        w_{0,1}& \ge \dfrac{\text{d}_1}{2} + a, \\[3pt]
        w_{0,i}&\ge a \quad \forall\, 1<i\le m_1.
    \end{aligned}
    \label{eq:interSCCsynchronizationconditions}
\end{equation}
Since the arc weights \( w_{0,i} \) are not subject to any upper bounds, a feasible solution always exists. Consequently, the required coupling condition for synchronization can always be achieved.
\end{proof}
Thus, assigning the inter-SCC arc weights to satisfy the above inequalities ensures synchronization of the SCC under consideration.
\begin{algorithm}[h]

\caption{Inter-SCC Arc Weight Allocation}
\label{alg:inter-scc-step}
\begin{algorithmic}
\State \textbf{Input:} $\G$ (with all the intra-SCC arc weights allocated) and $a>0$ (system parameter).
\State \textbf{Output:} $\G$ with inter-SCC arc weights assigned.
\State \textbf{Step 1:} Extract a non-source SCC - $\G_{k}$ that is unprocessed .
\State \textbf{Step 2:} Compute vertex imbalances of $\G_{k}$.
\State \textbf{Step 3:} Let $v_1$ be the vertex with positive vertex imbalance.
\State \textbf{Step 4:} Choose a vertex $v \in \G_{k}$ with incoming inter-SCC arcs.
\State \textbf{Step 5:} Let $I(v)$ be its incoming inter-SCC arcs in $\G$.
\State \textbf{Step 6:}
\textbf{if} $v = v_1$, distribute $a + \dfrac{d_1}{2}$ evenly among $I(v)$; \\
\hspace{1.15cm}\textbf{else} distribute $a$ evenly among $I(v)$.
\State \textbf{Step 7:} \textbf{if} vertices with unallocated inter-SCC incoming arcs remain in $S$, return to Step 5.
\State \textbf{Step 8:} \textbf{if} unprocessed non-source SCCs remain, return to Step 1.
\State \textbf{Step 9:} Return $\G$.
\end{algorithmic}
\end{algorithm}

The following are the key insights of theorem~\ref{thm:Arc weights for digraphs with trivial SCCs} :
\begin{itemize}
    \item Inter-SCC coupling strengths are determined by the choice of root vertex within each SCC.
    \item Vertices with higher internal connectivity (within the SCC) require lower coupling strengths.
    \item When multiple candidates exist, prioritizing those receiving more incoming inter-SCC arcs further reduces required coupling strengths.
    \item \textbf{Practical heuristic:} Select the root vertex that minimizes the ratio of sum of directed path lengths to the number of incoming inter-SCC arcs, thereby optimizing coupling-strength allocation.
\end{itemize}

\subsection{Example}
We consider a digraph with six vertices (Fig.~\ref{fig:connectivity_digraph}) to illustrate the proposed coupling strength allocation procedure. The digraph $\G$ consists of two SCCs: the source SCC $\G^{1}$ (blue vertices) and the non-source SCC $\G^{2}$ (white vertices). Condensing the source SCC $\G^{1}$ into a single vertex $v_0$ yields the reduced digraph $\tilde{\G}^{2}$.
\begin{figure}[h]
\centering
\begin{tikzpicture}[>=stealth,scale=0.3,
    redvertex/.style={circle, draw, fill=none, pattern color=black!50!red, minimum size=14pt, inner sep=0pt},
    bluevertex/.style={circle, draw, fill=blue!30, minimum size=14pt, inner sep=0pt},
    greenvertex/.style={circle, draw, fill=green!30, pattern=north west lines, pattern color=black!40!green, minimum size=14pt, inner sep=0pt},
    arrow/.style={thick},
    node distance=5mm]

\begin{scope}[shift={(-8,10)}]
    \foreach \i/\angle/\style in {
        1/90/redvertex,2/30/redvertex,3/-30/redvertex,
        4/-90/bluevertex,5/-150/bluevertex,6/150/bluevertex}
        \node[style=\style] (g1v\i) at (\angle:4cm) {\i};

    \node[above=of g1v1,yshift=-0.5cm] {$\G$};

    \draw[->,arrow] (g1v1) -- (g1v2);
    \draw[->,arrow] (g1v2) -- (g1v3);
    \draw[->,arrow] (g1v3) -- (g1v1);
    \draw[->,arrow] (g1v4) -- (g1v1);
    \draw[->,arrow] (g1v4) -- (g1v2);
    \draw[->,arrow] (g1v5) -- (g1v1);
    \draw[->,arrow] (g1v4) -- (g1v5);
    \draw[->,arrow] (g1v5) -- (g1v6);
    \draw[->,arrow] (g1v6) -- (g1v4);
\end{scope}

\begin{scope}[shift={(8,10)}]
    \foreach \i/\angle/\style in {
        1/90/redvertex,2/30/redvertex,3/-30/redvertex,0/-150/bluevertex}
        \node[style=\style] (g4v\i) at (\angle:4cm) {\i};
    \node[above=of g4v1,yshift=-0.5cm] {$\tilde{\G}^{2}$};    
    \draw[->,arrow] (g4v1) -- (g4v2);
    \draw[->,arrow] (g4v2) -- (g4v3);
    \draw[->,arrow] (g4v3) -- (g4v1);
    \draw[->,arrow] (g4v0) -- (g4v2);
    \draw[->,arrow] (g4v0) -- (g4v1);
\end{scope}
\end{tikzpicture}
\caption{Example digraph consisting of two SCCs, corresponding to the induced digraphs formed by the blue and white vertex sets. The digraph on the right depicts the reduced digraph obtained by contracting the source SCC into a single vertex.}
\label{fig:connectivity_digraph}
\end{figure}
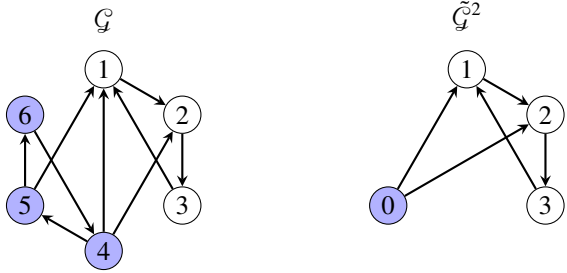

\subsubsection*{A. Intra-SCC Allocation}

For the source SCC $\G^{1}$, we fix the vertex 5 as the root vertex and compute the directed paths to every other vertex. The arcs along these directed paths are allocated decreasing weights and the corresponding arc-weight vectors are computed. Each such vector is scaled by $a$ and summed to obtain a negative imbalance arc-weight vector for $\G^{1}$ which is $\E^{-} = a[\,3,\,1,\,0\,]^\top$. Then, we find the directed cycles which in this case is a single cycle (\ref{tab:Gk1}). Since this is a source SCC, we find $\Delta w$ ($8a$) and scale the vector ($\E^{0} =a[\,8,\,8,\,8\,]^\top$ ) corresponding to the cycle to satisfy the synchronization condition. These two vectors are summed to compute the arc-weight vector for the source SCC i.e., $\E_{\text{sync}} = \E^{0} + \E^{-} = a[\,11,\,9,\,8\,]^\top.$

\begin{table}[h]
\caption{Directed cycles, directed paths and corresponding vectors for SCCs of $\G$}
\label{tab:Gk1}
\begin{tabular}{llll}
\hline
\textbf{SCC} & \textbf{Category} & \textbf{Cycle/Path} & \textbf{Vector}\\ \hline
\multirow{3}{*}{$\G^{1}$}
 & Directed Path  & $5 \rightarrow 6$ & $\left[\,1,\,0,\,0\,\right]^{\top}$ \\
 & Directed Path  & $5 \rightarrow 6 \rightarrow 4$ & $\left[\,2,\,1,\,0\,\right]^{\top}$ \\
 & Directed Cycle & $5 \rightarrow 6 \rightarrow 4 \rightarrow 5$ & $\left[\,1,\,1,\,1\,\right]^{\top}$ \\\hline
\multirow{3}{*}{$\G^{2}$}
 & Directed Path  & $1 \rightarrow 2$ & $\left[\,1,\,0,\,0\,\right]^{\top}$ \\
 & Directed Path  & $1 \rightarrow 2 \rightarrow 3$ & $\left[\,2,\,1,\,0\,\right]^{\top}$ \\
 & Directed Cycle & $1 \rightarrow 2 \rightarrow 3 \rightarrow 1$ & $\left[\,1,\,1,\,1\,\right]^{\top}$\\\hline
\end{tabular}
\end{table}

For the \textbf{non-source SCC} $\G^{2}$, for the non-trivial SCC $\G^{2}$, we fix vertex 1 as the root vertex and compute $\E^{-} = a[\,3,\,1,\,0\,]^\top$ as before and scale it by $a$ but not $\qquad
X^{0} = [\,1,\,1,\,1\,]^\top.$ since its only purpose is to ensure that arc $(3,1)$ has a nonzero weight. We add $X^{0}$ to $\E^{-}$, with zero-padding applied to arcs originating from vertex $v_0$. Excluding inter-SCC arcs, the arc-weight vector is $[\,0,\,0,\,3a+1,\,a+1,\,1\,]^\top$.

\subsubsection*{B. Inter-SCC Allocation}
We now determine the weights of arcs connecting the condensed vertex $v_0$ to $\G^{2}$. By theorem~\ref{thm:Arc weights for digraphs with trivial SCCs}, the required weights for arcs incident on $v_1$ are
\[
w_{01} = \frac{3a}{2} + a = 2.5a,
\qquad
w_{02} = a.
\]

Then, the weight $w_{01} = 2.5a$ assigned is equally divided between arcs $(4,1)$ and $(5,1)$.

This gives us the final arc-weight vector for $\G^{2}$ to be $[\,2.5a,\,a,\,3a+1,\,a+1,\,1\,]^\top.$

\begin{figure*}[h]
\centering
\includegraphics[width=0.8\linewidth]{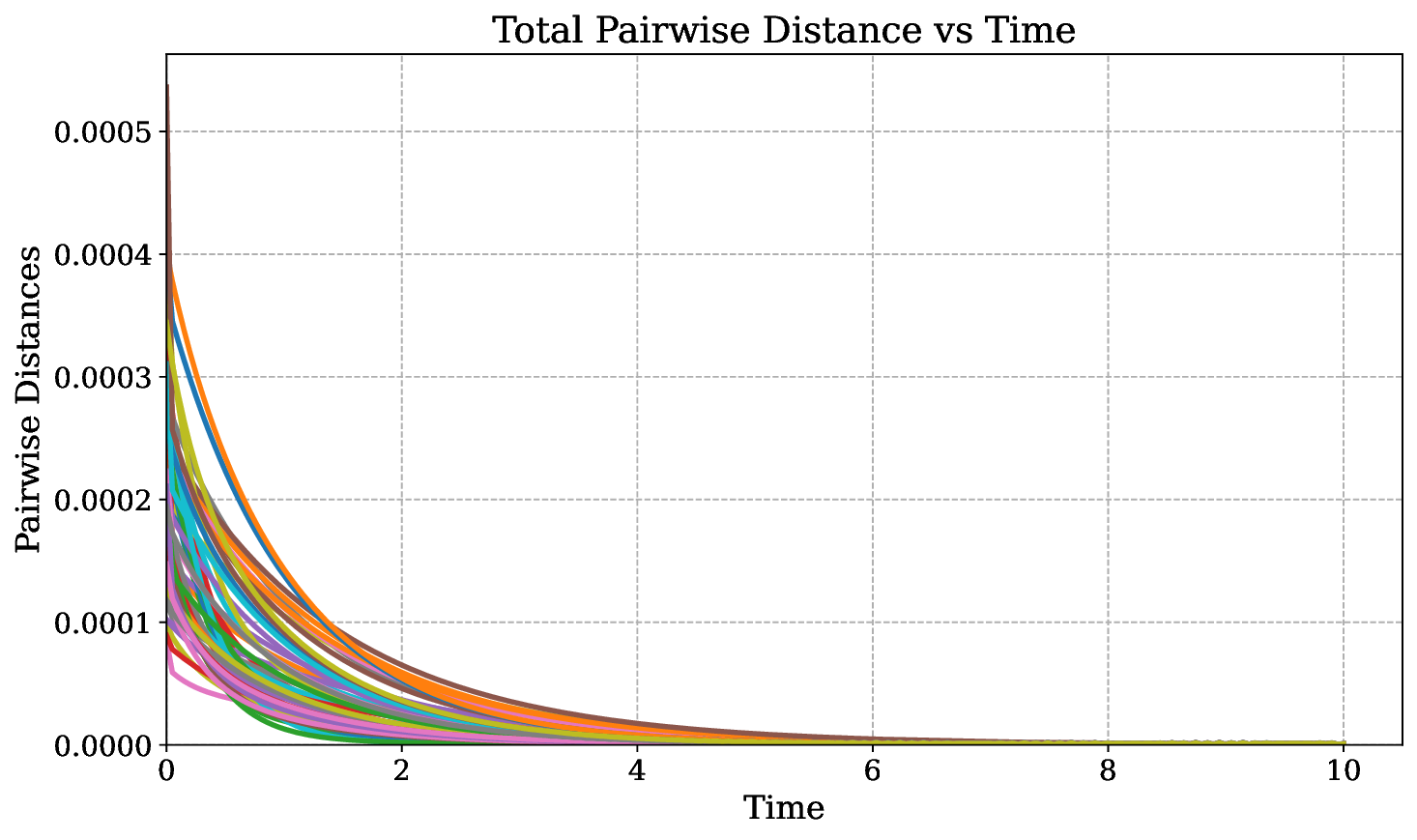}
\caption{Pairwise distances between all state pairs in a 100-node coupled Lorenz network. All distances rapidly converge toward zero.}
\label{fig:pairwise_distances}
\end{figure*}

\section{Simulation and Results}
We chose a network comprising of 100 identical Lorenz attractors with the following model :
\begin{equation}
  \dot{x} = \sigma (y - x), \;\dot{y} = x (r - z) - y,\; \dot{z} = x y - b z
\end{equation}
Here, $\sigma = 10$, $r = 28$, $b = 8/3$. We computed the system parameter $a$ using the formula $a = \dfrac{b(b+1)(r+\sigma)^{2}}{16(b-1)} -\sigma$ as directed in \cite{liu2013coupling}. Initial conditions were randomly sampled from a uniform distribution.

To verify that the selected coupling strengths lead to synchronization, we first extracted a directed spanning tree from the connectivity digraph of the network. We then computed the pairwise state differences for the dynamical systems corresponding to the vertices that were adjacent in this tree. Since a directed spanning tree guarantees reachability of all nodes from the root, convergence of the state differences along all tree arcs to zero implies global synchronization of the network. 

As illustrated in Fig.~\ref{fig:pairwise_distances}, the state differences corresponding to all spanning-tree edges decay asymptotically to zero, thereby confirming that the entire network achieves synchronization.

\subsection{Complexity}
For an SCC with $n'$ vertices and $m'$ arcs, the procedure involves computing $n'-1$ directed paths together with determining a directed-ear basis. This task has the same computational complexity as constructing a cycle basis via a spanning tree, which, as discussed in \cite{ryan1981comparison}, has complexity $O(n'\gamma)$ where $\gamma = m' - n' + 1$. 

In the worst-case scenario of a complete digraph, where $m '= \frac{n'(n'-1)}{2}$, we have $\gamma = O(n'^2)$, leading to an overall complexity of $O(n^3)$ for the connectivity digraph as a whole.

\section{Conclusion}
In this work, we presented a graph-theoretic framework for allocating coupling-strengths in a diffusively coupled dynamical system. We eliminated the need for inequality solvers and detailed system models and built on standard graph-traversal operations that are well established and computationally efficient. Since our method provides a principled approach for dealing with dynamical networks using primarily connectivity information, it is particularly useful for networks where obtaining precise system models may be impractical.

\medskip

Nevertheless, the scope of the present framework is defined by certain structural and computational assumptions. The proposed algorithm assumes that the connectivity digraph contains a directed spanning tree; otherwise, the digraph contains multiple source SCCs. In such cases, one possible remedy is to select one source SCC and connect it to the others through external inputs to enforce the required connectivity structure. The framework can, in principle, be extended to slowly time-varying topologies; however, the current implementation incurs a computational complexity of $O(n^3)$ per SCC, motivating the development of more efficient linear or near-linear algorithms and incremental update schemes. Finally, since the required coupling strengths scale with the system parameter $a$, budget constraints on coupling strengths may become restrictive, leading to the problem of selecting and strengthening arcs to synchronize the largest possible subset of oscillators.

\section*{Acknowledgment}
The authors acknowledge the use of language assistance tools, including ChatGPT, Gemini, DeepSeek, and Writefull for paraphrasing suggestions and grammar refinement during manuscript preparation. These tools were utilized solely for linguistic enhancement and not for the generation of any technical or scientific content.

The authors also thank Varatharajan Meenakshi Sundaram, a research scholar at the department of electrical engineering, Indian Institute of Technology Madras, for his valuable feedback and guidance. His careful review for technical errors and constructive suggestions significantly improved the quality and clarity of the manuscript.

\bibliographystyle{elsarticle-num} 
\bibliography{References_clean}

@article{liu2015synchronization,
  title={Synchronization in directed complex networks using graph comparison tools},
  author={Liu, Hui and Cao, Ming and Wu, Chai Wah and Lu, Jun-An and Tse, Chi K},
  journal={IEEE Transactions on Circuits and Systems I: Regular Papers},
  volume={62},
  number={4},
  pages={1185--1194},
  year={2015},
  publisher={IEEE}
}

@article{liu2013coupling,
  title={Coupling strength allocation for synchronization in complex networks using spectral graph theory},
  author={Liu, Hui and Cao, Ming and Wu, Chai Wah},
  journal={IEEE Transactions on Circuits and Systems I: Regular Papers},
  volume={61},
  number={5},
  pages={1520--1530},
  year={2013},
  publisher={IEEE}
}

@book{wu2007synchronization,
  title={Synchronization in complex networks of nonlinear dynamical systems},
  author={Wu, Chai Wah},
  year={2007},
  publisher={World scientific}
}

@article{belykh2004connection,
  title={Connection graph stability method for synchronized coupled chaotic systems},
  author={Belykh, Vladimir N and Belykh, Igor V and Hasler, Martin},
  journal={Physica D: nonlinear phenomena},
  volume={195},
  number={1-2},
  pages={159--187},
  year={2004},
  publisher={Elsevier}
}

@article{loebl2001some,
  title={Some remarks on cycles in graphs and digraphs},
  author={Loebl, Martin and Matamala, Martin},
  journal={Discrete mathematics},
  volume={233},
  number={1-3},
  pages={175--182},
  year={2001},
  publisher={Elsevier}
}

@article{ryan1981comparison,
  title={A comparison of three algorithms for finding fundamental cycles in a directed graph},
  author={Ryan, Doris R and Chen, Stephen},
  journal={Networks},
  volume={11},
  number={1},
  pages={1--12},
  year={1981},
  publisher={Wiley Online Library}
}

@article{mochizuki2013dynamics,
  title={Dynamics and control at feedback vertex sets. II: A faithful monitor to determine the diversity of molecular activities in regulatory networks},
  author={Mochizuki, Atsushi and Fiedler, Bernold and Kurosawa, Gen and Saito, Daisuke},
  journal={Journal of theoretical biology},
  volume={335},
  pages={130--146},
  year={2013},
  publisher={Elsevier}
}

@article{fiedler2013dynamics,
  title={Dynamics and control at feedback vertex sets. I: Informative and determining nodes in regulatory networks},
  author={Fiedler, Bernold and Mochizuki, Atsushi and Kurosawa, Gen and Saito, Daisuke},
  journal={Journal of Dynamics and Differential Equations},
  volume={25},
  number={3},
  pages={563--604},
  year={2013},
  publisher={Springer}
}

@article{belykh2006generalized,
  title={Generalized connection graph method for synchronization in asymmetrical networks},
  author={Belykh, Igor and Belykh, Vladimir and Hasler, Martin},
  journal={Physica D: Nonlinear Phenomena},
  volume={224},
  number={1-2},
  pages={42--51},
  year={2006},
  publisher={Elsevier}
}

@article{mirsky2009model,
  title={A model of the cell-autonomous mammalian circadian clock},
  author={Mirsky, Henry P and Liu, Andrew C and Welsh, David K and Kay, Steve A and Doyle III, Francis J},
  journal={Proceedings of the National Academy of Sciences},
  volume={106},
  number={27},
  pages={11107--11112},
  year={2009},
  publisher={National Academy of Sciences}
}

@article{Oda2005,
  author  = {Oda, K. and Matsuoka, Y. and Funahashi, A. and Kitano, H.},
  title   = {A comprehensive pathway map of epidermal growth factor receptor signaling},
  journal = {Mol. Syst. Biol.},
  year    = {2005},
  volume  = {1},
  pages   = {2005.0010},
  doi     = {10.1038/msb4100014},
  pmid    = {16729045},
  pmcid   = {PMC1681468}
}

@article{delellis2010quad,
  title={On QUAD, Lipschitz, and contracting vector fields for consensus and synchronization of networks},
  author={DeLellis, Pietro and di Bernardo, Mario and Russo, Giovanni},
  journal={IEEE Transactions on Circuits and Systems I: Regular Papers},
  volume={58},
  number={3},
  pages={576--583},
  year={2010},
  publisher={IEEE}
}

@article{russo2009contraction,
  title={Contraction theory and master stability function: Linking two approaches to study synchronization of complex networks},
  author={Russo, Giovanni and Di Bernardo, Mario},
  journal={IEEE Transactions on Circuits and Systems II: Express Briefs},
  volume={56},
  number={2},
  pages={177--181},
  year={2009},
  publisher={IEEE}
}

\end{document}